\documentclass[a4paper,USenglish]{llncs}

\usepackage{pdflscape}

\usepackage{thm-restate}
\usepackage{ntheorem}
\theoremstyle{plain}
\theorembodyfont{}
\theoremsymbol{}
\theoremprework{}
\theorempostwork{}
\theoremseparator{.}

\renewtheorem{theorem}{Theorem}
\renewtheorem{lemma}{Lemma}
\renewtheorem{claim}{Claim}
\usepackage{subfig}

\renewcommand{\epsilon}{\ensuremath{\varepsilon}}

\usepackage{amsmath,amssymb,gensymb}

\usepackage{lineno}
\let\doendproof\endproof
\renewcommand\endproof{~\hfill\qed\doendproof}

\usepackage{dsfont}
\usepackage{xspace}
\usepackage[colorlinks=true]{hyperref}
\usepackage{paralist}
\usepackage[color=yellow]{todonotes}
\usepackage{wrapfig}
\usepackage{floatrow}
\usepackage{thm-restate}
\usepackage{rotating}
\usepackage{algorithm}
\usepackage[noend]{algpseudocode}
\usepackage{stmaryrd}
\usepackage{fullpage}

\DeclareMathOperator{\xx}{\mathrm{x}}
\DeclareMathOperator{\yy}{\mathrm{y}}

\renewcommand{\paragraph}[1]{\smallskip\noindent\textbf{#1}\xspace}
\graphicspath{{img/}}

\usepackage[nameinlink,capitalise]{cleveref}
\Crefname{observation}{Observation}{Observations}
\Crefname{algorithm}{Algorithm}{Algorithms}
\Crefname{section}{Section}{Sections}
\Crefname{observation}{Observation}{Observations}
\Crefname{lemma}{Lemma}{Lemmas}
\Crefname{claim}{Claim}{Claims}
\Crefname{figure}{Fig.}{Figs.}
\Crefname{figure}{Fig.}{Figs.}
\Crefname{enumi}{Property}{Properties}

\newcommand{\remove}[1]{}

\usepackage{color}
\definecolor{realblue}{rgb}{0,0,1}
\definecolor{blue}{rgb}{0.274,0.392,0.666}
\definecolor{darkerblue}{rgb}{0.094,0.455,0.804}
\definecolor{darkblue}{rgb}{0.063,0.306,0.545}
\definecolor{red}{rgb}{0.627,0.117,0.156}
\definecolor{green}{rgb}{0,0.588,0.509}
\definecolor{orange}{rgb}{0.903,0.739,0.382}
\definecolor{realred}{rgb}{1,0,0}

\newcommand{\blue}[1]{{{\textcolor{darkblue}{#1}\xspace}}}

\newcommand{\darkblue}[1]{{{\textcolor{darkblue}{#1}\xspace}}}

\renewcommand{\emph}[1]{\darkblue{\em #1}}

\hypersetup{colorlinks,linkcolor={darkblue},citecolor={darkblue},urlcolor={darkblue}} 

\newcommand{\rome}{$^1$}
\newcommand{\ottawa}{$^2$}

\title{On the Area Requirements of Planar Greedy Drawings of Triconnected Planar Graphs\thanks{Partially supported by the MSCA-RISE project ``CONNECT'', N$^\circ$~734922, by the NSERC of Canada, and by the MIUR-PRIN project ``AHeAD'', N$^\circ$~20174LF3T8.}}
\author{
Giordano {Da Lozzo}\rome,
Anthony {D'Angelo}\ottawa, 
\and
Fabrizio Frati\rome
}
\institute{
\rome Roma Tre University, Rome, Italy $\cdot$
\href{mailto:giordano.dalozzo@uniroma3.it,fabrizio.frati@uniroma3.it}{firstname.lastname@uniroma3.it}\\
\ottawa Carleton University, Ottawa, Canada $\cdot$ \href{mailto:anthonydangelo@cmail.carleton.ca}{anthonydangelo@cmail.carleton.ca}
}

\begin{document}
\maketitle

\begin{abstract}
In this paper we study the area requirements of planar greedy drawings of triconnected planar graphs. Cao, Strelzoff, and Sun exhibited a family $\cal H$ of subdivisions of triconnected plane graphs and claimed that every planar greedy drawing of the graphs in $\mathcal H$ respecting the prescribed plane embedding requires exponential area. However, we show that every $n$-vertex graph in $\cal H$ actually has a planar greedy drawing respecting the prescribed plane embedding on an $O(n)\times O(n)$ grid. This reopens the question whether triconnected planar graphs admit planar greedy drawings on a polynomial-size grid. Further, we provide evidence for a positive answer to the above question by proving that every $n$-vertex Halin graph admits a planar greedy drawing on an $O(n)\times O(n)$ grid. Both such results are obtained by actually constructing drawings that are convex and angle-monotone. Finally, we consider $\alpha$-Schnyder drawings, which are angle-monotone and hence greedy if $\alpha\leq 30\degree$, and show that there exist planar triangulations for which every $\alpha$-Schnyder drawing with a fixed $\alpha<60\degree$ requires exponential area for any resolution rule.  
\end{abstract}


\section{Introduction}
Let $(M,d)$ be a geometric metric space, where $M$ is a set of points and $d$ is a metric on $M$. A \emph{greedy embedding} of a graph $G$ into $(M,d)$ is a function $\phi$ that maps each vertex $v$ of $G$ to a point $\phi(v)$ in $M$ in such a way that, for every ordered pair $(u,v)$ of vertices of $G$, there is a distance-decreasing path from $u$ to $v$ in $G$, i.e., a path $(u=w_1,w_2,\dots,w_k=v)$ such that $d\big(\phi(w_i),\phi(v)\big)>d(\phi\big(w_{i+1}),\phi(v)\big)$, for $i=1,\dots,k-1$. Greedy embeddings, introduced by Rao et al.~\cite{rpss-grwli-03}, support a simple and local routing scheme, called \emph{greedy routing}, in which a vertex forwards a packet to any neighbor that is closer to the packet's destination than itself.
In order for greedy routing to be efficient, a greedy embedding should be \emph{succinct}, i.e., a polylogarithmic number of bits should be used to store the coordinates of each vertex. A number of algorithms have been proposed to construct succinct greedy embeddings of graphs~\cite{eg-sggruhg-11,gs-sggrep-09,hz-osgdptt-14,m-dgra-07,DBLP:journals/comcom/SunZFZ17,wh-sscgd3pg-14}. Notably, every graph admits a succinct greedy embedding into the hyperbolic plane~\cite{eg-sggruhg-11}. A natural choice is the one of considering $M$ to be the Euclidean plane $\mathbb R^2$ and $d$ to be the Euclidean distance $\ell_2$. Within this setting, not every graph~\cite{pr-crgr-05,DBLP:conf/algosensors/PapadimitriouR04}, and not even every binary tree~\cite{lm-srgems-10,np-egdt-13}, admits a greedy embedding; further, there exist trees whose every greedy embedding requires a polynomial number of bits to store the coordinates of some of the vertices~\cite{adf-sgd-12}.

From a theoretical point of view, most research efforts have revolved around two conjectures posed by Papadimitriou and Ratajczak~\cite{pr-crgr-05,DBLP:conf/algosensors/PapadimitriouR04}. The first one asserts that every $3$-connected planar graph admits a \emph{greedy drawing}, i.e., a straight-line drawing in $\mathbb R^2$ that induces a greedy embedding into $(\mathbb R^2,\ell_2)$. This conjecture has been confirmed independently by Leighton and Moitra~\cite{lm-srgems-10} and by Angelini et al.~\cite{afg-acgdt-10}. The second conjecture, which strengthens the first one, asserts that every $3$-connected planar graph admits a greedy drawing that is also convex. While this conjecture is still open, it has been recently proved by the authors of this paper that every $3$-connected planar graph admits a {\em planar} greedy drawing~\cite{ddf-opgdtpg-18}. 

An interesting question is whether succinctness and planarity can be achieved simultaneously. That is, does every $3$-connected planar graph admit a planar, and possibly convex, greedy drawing on a polynomial-size grid? Cao, Strelzoff, and Sun~\cite{css-osggrep-09} claimed a negative answer by exhibiting a family~$\cal H$ of subdivisions of $3$-connected plane graphs and by showing that, for any $n$-vertex graph in~$\cal H$, any planar greedy drawing that respects the prescribed plane embedding requires $2^{\Omega(n)}$ area and hence~$\Omega(n)$ bits for representing the coordinates of some vertices. 

Subsequently to the definition of greedy drawings, a number of more constrained graph drawing standards have been introduced and studied. Analogously to greedy drawings, they all concern straight-line drawings in $\mathbb R^2$. In a \emph{self-approaching} drawing~\cite{acglp-sag-12,dfg-icgps-15,npr-osaicd-16}, for every pair of vertices $u$ and $v$, there is a self-approaching path from $u$ to $v$, i.e., a path $P$ such that $\ell_2(a,c)>\ell_2(b,c)$, for any three points $a$, $b$, and $c$ in this order along $P$. In an \emph{increasing-chord} drawing~\cite{acglp-sag-12,dfg-icgps-15,npr-osaicd-16}, for every pair of vertices $u$ and $v$, there is a path from $u$ to $v$ which is self-approaching both from $u$ to $v$ and from $v$ to $u$. In an \emph{angle-monotone} drawing~\cite{bbcklv-gt-16,dfg-icgps-15,DBLP:conf/cccg/LubiwO17,lm-clr-19}, for every pair of vertices $u$ and $v$, there exists a $\beta$-monotone path from $u$ to $v$ for some angle $\beta$, i.e., a path $P=(w_1=u,w_2,\dots,w_k=v)$ such that, for each $i=1,\dots,k-1$, the edge $(w_i,w_{i+1})$ lies in the closed $90\degree$-wedge centered at $w_i$ and bisected by the ray originating at $w_i$ with slope $\beta$. Note that an angle-monotone drawing is increasing-chord, an increasing-chord drawing is self-approaching, and a self-approaching drawing is greedy. The first implication was proved in~\cite{dfg-icgps-15}, while the other two descend from the definitions. Finally, a notable class of planar straight-line drawings are \emph{$\alpha$-Schnyder} drawings~\cite{npr-osaicd-16}, which are angle-monotone if $\alpha\leq 30\degree$ and will be formally defined later.

{\bf Our contributions.} We show that every $n$-vertex graph in the family $\cal H$ defined by Cao et al.~\cite{css-osggrep-09} actually admits a convex angle-monotone drawing that respects the prescribed plane embedding and that lies on an $O(n)\times O(n)$ grid. This refutes their claim that every planar greedy drawing of an $n$-vertex graph in $\cal H$ requires~$\Omega(n)$ bits for representing the coordinates of some vertices and reopens the question about the existence of succinct planar greedy drawings of $3$-connected planar graphs. Further, we provide an indication that this question might have a positive answer by proving that the $n$-vertex Halin graphs, a notable family of triconnected planar graphs, admit convex angle-monotone drawings on an $O(n)\times O(n)$ grid. Finally, we show that there exist bounded-degree planar triangulations whose every $\alpha$-Schnyder drawing requires exponential area, for any fixed $\alpha<60\degree$. This result was rather surprising to us, as any planar triangulation admits a $60\degree$-Schnyder drawing on an $O(n) \times O(n)$ grid~\cite{DBLP:conf/soda/Schnyder90}; further, although $30\degree$-Schnyder drawings have been proved to exist for all stacked triangulations, our result shows that they are not the right tool to obtain succinct planar greedy drawings.



\section{Definitions and Preliminaries}\label{se:preliminaries}

A \emph{straight-line drawing} of a graph maps each vertex to a point in the plane and each edge to a straight-line segment between its end-points. A drawing is \emph{planar} if no two edges cross. A planar drawing partitions the plane into connected regions, called \emph{faces}. The only unbounded face is the \emph{outer face}; the other faces are \emph{internal}. Two planar drawings of the same connected planar graph are \emph{equivalent} if they determine the same circular order of the edges incident to each vertex. A \emph{planar embedding} is an equivalence class of planar drawings.
A \emph{plane graph} is a planar graph equipped with a planar embedding and a designated outer face.
A straight-line drawing is \emph{convex} if it is planar and every face is delimited by a convex polygon. A \emph{grid drawing} is such that each vertex is mapped to a point with integer coordinates. The \emph{width} (resp.\ \emph{height}) of a grid drawing is the number of grid columns (rows) intersecting it. We say that a drawing \emph{lies on a $W\times H$ grid} if it is a grid drawing with width $W$ and height $H$. The \emph{area} of a graph drawing is usually defined as the area of the smallest axis-parallel rectangle enclosing the drawing (when proving upper bounds) or as the area of the smallest convex polygon enclosing the drawing (when proving lower bounds). Any constraint implying a finite minimum area for a graph drawing is called a \emph{resolution rule}. 

From here on out, we measure angles in radians. In a straight-line drawing of a graph, the \emph{slope} of an edge $(u,v)$ is the angle spanned by a counter-clockwise rotation around $u$ of a ray originating at $u$ and directed rightwards bringing the ray to overlap with $(u,v)$; hence, the edge slopes are in the range $[0,2\pi)$. We denote by $(\xx(v),\yy(v))$ the point in the plane representing a vertex $v$ in a drawing of a graph.

A \emph{planar triangulation} $G$ is a plane graph whose every face is bounded by a $3$-cycle. Denote by $(a_1,a_2,a_3)$ the $3$-cycle bounding the outer face of $G$.
A \emph{Schnyder wood} $(\mathcal T_1,\mathcal T_2, \mathcal T_3)$ of $G$ is an assignment of directions and colors $1$, $2$ and~$3$ to the internal edges of $G$ such that the following two properties hold; see the figure below and refer to~\cite{DBLP:conf/soda/Schnyder90}. \mbox{Let $i-1 = 3$, if $i=1$, and let $i+1=1$, if $i=3$.}

\begin{description}
	\item[Property (1)] Each internal vertex $v$ has one outgoing edge $e_i$ of each color $i$, with $i = 1, 2, 3$. The outgoing edges $e_1$, $e_2$, and $e_3$ appear in this clockwise order at $v$. Further, all the incoming edges of color $i$ appear in the clockwise sector between the edges $e_{i+1}$ and $e_{i-1}$. 
	\item[Property (2)] At the external vertex $a_i$, all the internal edges are incoming and of color $i$. 
	\begin{center}
	\begin{tabular}{ccc}
		\centering
			\includegraphics[page=1]{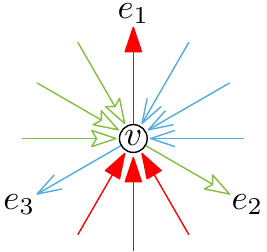}
			&~~~~~~~~~~~~~~~~&
			\includegraphics[page=2]{schnyder-conditions}
			\\
			Property (1) 			
			&~~~~~~~~~~~~~~~~& 
			Property (2)
	\end{tabular}
	\end{center}
\end{description}
		

\noindent
For $0 < \alpha \leq \frac{\pi}{3}$, 
a planar straight-line drawing of $G$ is an \emph{$\alpha$-Schnyder drawing} if, for
each internal vertex $v$ of $G$, 
its outgoing edge in $\mathcal T_1$ has direction in $[\frac{\pi}{2} - \frac{\alpha}{2} , \frac{\pi}{2} + \frac{\alpha}{2} ]$,  
its outgoing edge in $\mathcal T_2$ has direction in $[\frac{11\pi}{6} -\frac{\alpha}{2},\frac{11\pi}{6} + \frac{\alpha}{2}]$, and 
its outgoing edge in $\mathcal T_3$ has direction in $[\frac{7\pi}{6} - \frac{\alpha}{2},\frac{7\pi}{6} + \frac{\alpha}{2}]$. Observe that, by definition, in an $\alpha$-Schnyder drawing, for each internal vertex $v$ of $G$, 
its incoming edges in $\mathcal T_1$, if any, have direction in $[\frac{3\pi}{2} - \frac{\alpha}{2} , \frac{3\pi}{2} + \frac{\alpha}{2} ]$,
its incoming edges in $\mathcal T_2$, if any, have direction in $[\frac{5\pi}{6} - \frac{\alpha}{2} , \frac{5\pi}{6} + \frac{\alpha}{2} ]$, and
its incoming edges in $\mathcal T_3$, if any, have direction in $[\frac{\pi}{6} - \frac{\alpha}{2} , \frac{\pi}{6} + \frac{\alpha}{2} ]$.
\cref{fig:wedges} shows the angular widths of an $\alpha$-Schnyder drawing. ``Usual'' Schnyder drawings~\cite{DBLP:conf/soda/Schnyder90} are $60\degree$-Schnyder drawings; see, e.g.,~\cite{d-gdt-10}.

\section{Angle-Monotone Drawings of Cao-Strelzoff-Sun Graphs}\label{sse:Cao-Strelzoff-Sun}
Cao et al.~\cite{css-osggrep-09} defined the following family $\cal H$ of plane graphs. For every integer $i \geq 1$, the plane graph $\mathfrak{H}_i  \in \mathcal H$ on $3i+4$ vertices is inductively defined as follows:
\begin{wrapfigure}[6]{R}{6cm}
	\includegraphics[page=5,scale=0.65]{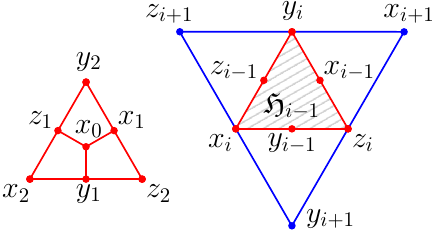}
	\label{fig:ao-Strelzoff-Sun-graphs}
\end{wrapfigure}
\begin{itemize}
\item The plane graph $\mathfrak{H}_1$ is composed of a cycle $(x_2,z_1,y_2,x_1,z_2,y_1)$ and of a vertex $x_0$ embedded inside such a cycle and adjacent to $x_1$, $y_1$, and $z_1$; see the \mbox{left part of the figure.}
\item For $i\geq 2$, the plane graph $\mathfrak{H}_i$ is obtained by embedding in the outer face of $\mathfrak{H}_{i-1}$ the vertices $x_{i+1}$, $y_{i+1}$, and $z_{i+1}$, and the edges of the cycle $(x_{i+1},z_{i},y_{i+1},x_{i},z_{i+1},y_{i})$, which bounds the outer face of $\mathfrak{H}_i$; see the right part of the figure.
\end{itemize}

In contrast to the result in~\cite{css-osggrep-09}, we prove the following.\footnote{The flaw in the proof presented in~\cite{css-osggrep-09} seems to be in the statement ``it is not difficult to see that the most economic way (i.e., consuming the minimum area) of stretching [$\mathfrak{H}_i$] into a greedy embedding is to do it symmetrically [$\dots$]''.}

\begin{theorem}\label{th:counter-couter-example}
Every $n$-vertex plane graph in $\mathcal H$ admits a planar angle-monotone drawing on an $O(n) \times O(n)$ grid that respects the plane embedding.
\end{theorem}

\begin{proof}
In order to prove the statement, we construct, for every $i \geq 1$, a planar straight-line drawing $\Gamma_i$ of $\mathfrak{H}_i = (V_i,E_i)$ satisfying the following properties:
\begin{enumerate}[\bf (i)]
\item \label{prop:grid} the vertices of $\mathfrak{H}_i$ lie on an $(2i+3) \times (2i+3)$ grid;
\item \label{prop:threepaths} there exist paths $p_i(\alpha)$, with $\alpha \in \{\frac{\pi}{2},\frac{5\pi}{4},\frac{7\pi}{4}\}$, originating at $x_0$ and each terminating at a distinct vertex in $\{x_{i+1},y_{i+1},z_{i+1}\}$, that are vertex-disjoint except at $x_0$, that together span all the vertices in $V_i$, and such that all the edges in $p_i(\alpha)$ have slope $\alpha$.
\end{enumerate}

\cref{prop:threepaths} implies that $\Gamma_i$ is angle-monotone. Namely, consider any two vertices $u$ and $v$ of $\mathfrak{H}_i$.
If both $u$ and $v$ belong to the same path $p_i(\alpha)$, then the subpath of $p_i(\alpha)$ from $u$ to $v$ is either $\alpha$-monotone or $(\pi+\alpha)$-monotone.
If $u \in p_i(\alpha)$ and $v \in p_i(\beta)$, with $\alpha \neq \beta$, then the path $p^*$ consisting of the subpath of $p_i(\alpha)$ from $u$ to $x_0$ and of the subpath of $p_i(\beta)$ from $x_0$ to $v$ is $\frac{\pi}{2}$-monotone (if $\beta=\frac{\pi}{2}$), or $\frac{3\pi}{2}$-monotone (if $\alpha=\frac{\pi}{2}$), or $\pi$-monotone (if $\alpha=\frac{7\pi}{4}$ and $\beta=\frac{5\pi}{4}$), or $0$-monotone (if $\alpha=\frac{5\pi}{4}$ and $\beta=\frac{7\pi}{4}$).

%
Our proof is by induction on $i$.

\textbf{Base case}. If $i=1$, we construct a drawing $\Gamma_1$ of $\mathfrak{H}_1$ as follows; refer to \cref{fi:cao-draw-base}. 
We place the vertex $x_0$ at the point $(0,0)$, the vertices $x_1$, $y_1$, and $z_1$ at the points $(0,1)$, $(1,-1)$, and $(-1,-1)$, respectively, and the vertices $x_2$, $y_2$, and~$z_2$ at the points $(-2,-2)$, $(0,2)$, and $(2,-2)$, respectively, and draw the edges of~$\mathfrak{H}_1$ as straight-line segments. 
By construction, $\Gamma_1$ is a planar straight-line grid drawing of $\mathfrak{H}_1$ on the 
$5 \times 5$ grid, thus satisfying \cref{prop:grid}.
Further, the paths 
$p_1(\frac{\pi}{2}) = (x_0, x_1, y_2)$,  
$p_1(\frac{5\pi}{4}) = (x_0, z_1, x_2)$, and
$p_1(\frac{7\pi}{4}) = (x_0, y_1, z_2)$ show that \cref{prop:threepaths} is satisfied by $\Gamma_1$.

\begin{figure}[tb]
	\subfloat[]{
		\includegraphics[page=16,scale=0.6]{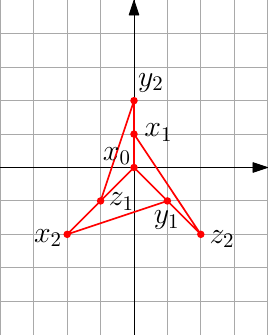}
		\label{fi:cao-draw-base}
	}
	\subfloat[]{
		\includegraphics[page=15,scale=0.6]{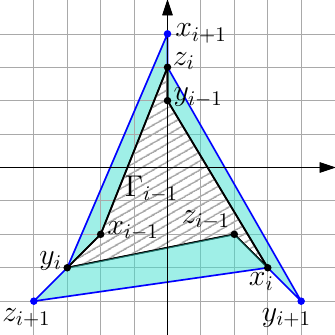}
		\label{fi:cao-draw-quadrilaterals}
	}
	\subfloat[]{
		\includegraphics[page=17,scale=0.6]{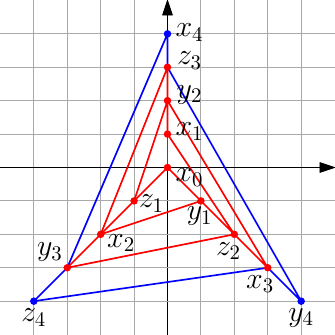}
		\label{fi:convex-cao-quasi}
	}
	\subfloat[]{
		\includegraphics[page=18,scale=0.6]{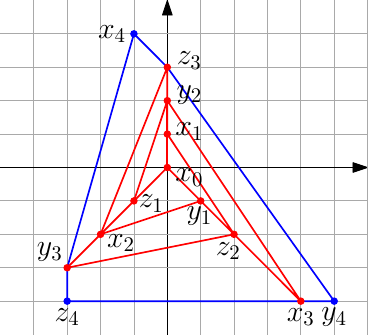}
		\label{fi:convex-cao}
	}
	\caption{
		Illustrations for the proof of \cref{th:counter-couter-example}:
		(a) The drawing $\Gamma_1$ of $\mathfrak{H}_1$;
		(b) The drawing $\Gamma_i$ of $\mathfrak{H}_i$ obtained from the drawing $\Gamma_{i-1}$ of $\mathfrak{H}_{i-1}$, with $i>1$.
		(d) The convex angle-monotone drawing $\Gamma'_3$ of $\mathfrak{H}_3$ obtained from (c) $\Gamma_3$.}
	\label{fi:cao}
\end{figure}

\textbf{Inductive case}. If $i>1$, suppose to have inductively constructed a drawing $\Gamma_{i-1}$ of $\mathfrak{H}_{i-1}$ satisfying \cref{prop:grid,prop:threepaths}. Assume, as in \cref{fi:cao-draw-quadrilaterals}, that $z_{i}$ is in~$p_{i-1}(\frac{\pi}{2})$, that $y_{i}$ is in $p_{i-1}(\frac{5\pi}{4})$, and that $x_{i}$ is in $p_{i-1}(\frac{7\pi}{4})$; the other cases can be treated analogously. 

We obtain $\Gamma_i$ from $\Gamma_{i-1}$ by placing $x_{i+1}$ at the point $(\xx(z_i),\yy(z_i)+1)$, $y_{i+1}$ at the point $(\xx(x_i)+1,\yy(x_i)-1)$, and $z_{i+1}$ at the point $(\xx(y_i)-1,\yy(y_i)-1)$, and by drawing the edges incident to these vertices as straight-line segments. We have the following.

\begin{claim} \label{cl:cao}
	$\Gamma_i$ satisfies \cref{prop:grid,prop:threepaths}.
\end{claim}

\begin{proof}
	First observe that, since $\Gamma_{i-1}$ is a grid drawing, by induction, we have that $x_{i+1}$, $y_{i+1}$, and $z_{i+1}$ have integer coordinates, hence $\Gamma_i$ is a grid drawing as well. By \cref{prop:threepaths} of $\Gamma_{i-1}$, all the vertices of $\mathfrak{H}_{i-1}$ lie on the straight-line segments connecting $x_0$ with $x_i$, $y_i$, and $z_i$. Hence, $\Gamma_{i-1}$ lies inside the triangle $\Delta_{i-1}$ with vertices $x_i$, $y_i$, and $z_i$. On the other hand, the edges incident to  $x_{i+1}$, $y_{i+1}$, and~$z_{i+1}$ lie in the exterior of $\Delta_{i-1}$ in $\Gamma_i$ (except, possibly, for their endpoints); since these edges do not cross each other, we have that $\Gamma_i$ is planar.
	
	We prove that $\Gamma_i$ satisfies \cref{prop:grid}. By construction, $\Gamma_i$ intersects two more grid rows and two more grid columns than $\Gamma_{i-1}$, hence it lies on the $(2i+3) \times (2i+3)$ grid, since $\Gamma_{i-1}$ lies on the $(2i+1) \times (2i+1)$ grid, by induction. 
	
	We prove that $\Gamma_i$ satisfies \cref{prop:threepaths}. Define the paths $p_i(\frac{\pi}{2})=p_{i-1}(\frac{\pi}{2})\cup (z_i,x_{i+1})$,  $p_i(\frac{5\pi}{4})=p_{i-1}(\frac{5\pi}{4})\cup (y_i,z_{i+1})$, and
	$p_i(\frac{7\pi}{4})=p_{i-1}(\frac{7\pi}{4})\cup (x_i,y_{i+1})$. We have that $p_i(\frac{\pi}{2})$, $p_i(\frac{5\pi}{4})$, and $p_i(\frac{7\pi}{4})$ originate at $x_0$ and are vertex-disjoint except at $x_0$, since $p_{i-1}(\frac{\pi}{2})$, $p_{i-1}(\frac{5\pi}{4})$, and $p_{i-1}(\frac{7\pi}{4})$  satisfy the same properties, by induction. By construction, each of $p_i(\frac{\pi}{2})$,  $p_i(\frac{5\pi}{4})$, and $p_i(\frac{7\pi}{4})$ terminates at a distinct vertex in $\{x_{i+1},y_{i+1},z_{i+1}\}$. The paths $p_i(\frac{\pi}{2})$, $p_i(\frac{5\pi}{4})$, and $p_i(\frac{7\pi}{4})$ together span all the vertices in $V_i$ since $p_{i-1}(\frac{\pi}{2})$, $p_{i-1}(\frac{5\pi}{4})$, and $p_{i-1}(\frac{7\pi}{4})$ together span all the vertices in~$V_{i-1}$ and $V_i=V_{i-1}\cup \{x_{i+1},y_{i+1},z_{i+1}\}$. Finally, the edges of $p_i(\frac{\pi}{2})$, $p_i(\frac{5\pi}{4})$, and~$p_i(\frac{7\pi}{4})$ have slope $\frac{\pi}{2}$, $\frac{5\pi}{4}$, and $\frac{7\pi}{4}$ in $\Gamma_i$, respectively, since by induction the edges of $p_{i-1}(\frac{\pi}{2})$, $p_{i-1}(\frac{5\pi}{4})$, and $p_{i-1}(\frac{7\pi}{4})$ have slope $\frac{\pi}{2}$, $\frac{5\pi}{4}$, and $\frac{7\pi}{4}$ in $\Gamma_{i-1}$, respectively, and since, by construction, the edges $(z_i,x_{i+1})$, $(y_i,z_{i+1})$, and $(x_i,y_{i+1})$ have slope $\frac{\pi}{2}$, $\frac{5\pi}{4}$, and $\frac{7\pi}{4}$ in $\Gamma_i$, respectively.
\end{proof}

\cref{cl:cao} concludes the induction and the proof of the theorem.\end{proof}

We note that, for $i\geq 1$, the graph $\mathfrak{H}_i$ even admits a {\em convex} angle-monotone drawing $\Gamma'_i$ on an $(2i+3) \times (2i+3)$ grid; indeed, $\Gamma'_i$ can be obtained from the planar angle-monotone drawing $\Gamma_i$ of $\mathfrak{H}_i$ described in the proof of \cref{th:counter-couter-example} by moving $x_i$ one unit to the right and one unit down, $y_{i+1}$ and $z_{i+1}$ one unit to the right, and $x_{i+1}$ one unit to the left; see \cref{fi:convex-cao-quasi,fi:convex-cao}. We have the following.

\begin{claim}
$\Gamma'_i$ is a convex angle-monotone\mbox{ drawing of $\mathfrak{H}_i$ on an $(2i+3) \times (2i+3)$ grid.}
\end{claim}

\begin{proof}
It is easy to see that $\Gamma'_i$ is a convex drawing of $\mathfrak{H}_i$ on an $(2i+3) \times (2i+3)$~grid. We prove that $\Gamma'_i$ is angle-monotone. Consider the following three paths: 
\begin{itemize} 
	\item $P_1=p_{i}(\frac{\pi}{2})\cup p_{i}(\frac{5\pi}{4})$; 
	\item $P_2=p_{i}(\frac{\pi}{2})\cup p_{i}(\frac{7\pi}{4})$; and
	\item $P_3=p_{i-1}(\frac{5\pi}{4})\cup p_{i}(\frac{7\pi}{4})$.
\end{itemize} 
If $u$ and $v$ both belong to the path $P_1$, $P_2$, or $P_3$, then the subpath of such a path from~$u$ to $v$ is  
\begin{itemize} 
	\item $\frac{\pi}{2}$-monotone or $\frac{3\pi}{2}$-monotone, 
	\item $\frac{3\pi}{4}$-monotone or $\frac{7\pi}{4}$-monotone, or
	\item $0$-monotone or $\pi$-monotone, respectively.
\end{itemize}
Note that $u$ and~$v$ both belong to one of $P_1$, $P_2$, or $P_3$, unless one of them, say $u$, is $z_{i+1}$ and the other one, say $v$, belongs to $p_{i}(\frac{7\pi}{4})$. In such a case, a $\beta$-monotone path $P$ from $u$ to $v$ can be defined as follows. If $v=x_i$, then $P$ coincides with the edge $(z_{i+1},x_i)$; if $v=y_{i+1}$, then $P$ coincides with the path $(z_{i+1},x_i,y_{i+1})$; in both cases, $P$ is $0$-monotone. Finally, if $v$ belongs to $p_{i-2}(\frac{7\pi}{4})$, then $P$ is defined as the subpath of $p_{i}(\frac{5\pi}{4})$ from $u$ to the only neighbor of $v$ in $p_{i}(\frac{5\pi}{4})$, and from that neighbor to~$v$; then $P$ is $\frac{\pi}{4}$-monotone. 
\end{proof}

He and Zhang~\cite{hz-osgdptt-14} pointed out that, although the graphs $\mathfrak{H}_i$'s are not \mbox{$3$-connected}, they can be made so by adding the three additional edges $(x_{i+1},y_{i+1})$, $(y_{i+1},z_{i+1})$, and $(z_{i+1},x_{i+1})$. Let $\mathfrak{H}^+_i$ be the resulting graph. We note here that the drawing $\Gamma_i$ of $\mathfrak{H}_i$ whose construction is described in the proof of \cref{th:counter-couter-example} can be turned into a convex angle-monotone drawing $\Gamma^+_i$ of $\mathfrak{H}^+_i$ simply by drawing the edges $(x_{i+1},y_{i+1})$, $(y_{i+1},z_{i+1})$, and $(z_{i+1},x_{i+1})$ \mbox{as straight-line segments.}

\section{Angle-Monotone Drawings of Halin Graphs}\label{sse:halin}

In this section, we show how to construct convex angle-monotone drawings of Halin graphs on a polynomial-size grid. 


We denote the number of leaves of a tree $T$ by $\ell(T)$. A tree whose all vertices but one are leaves is a \emph{star}. A \emph{rooted tree} $T$ is a tree with one distinguished vertex, called \emph{root} and denoted by $r(T)$. The \emph{height} of a rooted tree is the maximum number of edges in any path from the root to a leaf. In a rooted tree $T$, we denote by $T(v)$ the subtree of $T$ rooted at a vertex $v$. An \emph{ordered rooted tree} is a rooted tree in which the children of each internal vertex $u$ are assigned a left-to-right order $u_1,\dots,u_k$; the vertices $u_1$ and $u_k$ are the \emph{leftmost} and the \emph{rightmost child} of $u$, respectively.
The \emph{leftmost path}  of an ordered rooted tree $T$ is the path $(v_1,\dots,v_h)$ in $T$ such that $v_1$ is the root of $T$, $v_{i+1}$ is the leftmost child of $v_i$, for $i=1,\dots,h-1$, and $v_h$ is a leaf, which is called the \emph{leftmost leaf} of $T$. The \emph{rightmost path} and the \emph{rightmost leaf} of $T$ can be defined analogously.

A \emph{Halin graph} $G$ is a $3$-connected planar graph that admits a plane embedding $\mathcal E$ such that, by removing all the edges incident to the outer face $f_{\mathcal E}$ of $\mathcal E$, one gets a tree $T_G$ whose internal vertices have degree at least $3$ and whose leaves are incident to $f_{\mathcal E}$. We have the following main result.


\begin{theorem}\label{th:halin}
	Every $n$-vertex Halin graph $G$ admits a convex angle-monotone drawing on an $O(n) \times O(n)$ grid.
\end{theorem}

If $T_G$ contains one internal vertex, then $G$ is a wheel and a convex angle-monotone drawing on a $3 \times (n-1)$ grid can easily be computed; refer to \cref{fi:wheel}\blue{(top)}. In the following, we assume that $T_G$ contains at least two internal vertices. 

Let $\xi$ be an internal vertex of $T_G$ whose every neighbor is a leaf, except for one, which we denote by $\rho$; see \cref{fi:halin-a}. Such a vertex exists by the above assumption. Further, let $T \subset T_G$ be the tree obtained from $T_G$ by removing $\xi$ and all its adjacent leaves and by rooting the resulting tree at $\rho$. Also, let $S \subset T_G$ be the star obtained from $T_G$ by removing the vertices of $T$ and by rooting the resulting tree at $\xi$.
We regard $T$ and $S$ as ordered rooted trees such that the left-to-right order of the children of each vertex is the one induced by the plane embedding~$\mathcal E$ of $G$.
For any subtree $T'\subseteq T_G$, let $G[T']$ be the subgraph of $G$ induced by the vertices of $T'$.
In \cref{le:halin-tree}, we show how to construct a drawing $\Gamma$ of $G[T]$. Then, we will exploit \cref{le:halin-tree} in order to prove \cref{th:halin}.

\begin{lemma}\label{le:halin-tree}
The graph $G[T]$ has a drawing $\Gamma$ satisfying the following properties:
\begin{enumerate}[(i)]
	\item \label{prop:greedy} $\Gamma$ is angle-monotone and convex; 
	\item \label{prop:convex-grid} $\Gamma$ lies on a $W_\Gamma\times H_\Gamma$ grid, where $W_\Gamma = 2\ell(T)-1$ and $H_\Gamma = \ell(T)$;
	\item \label{prop:points} the leaves of $T$ lie at $(0,0),(2,0),\dots,(2\ell(T)-2,0)$, where the $i$-th leaf of $T$ lies at $(2i-2,0)$, for $i=1,\dots,\ell(T)$; and
	\item \label{prop:paths} for each vertex $v$ of $T$, the edges of the leftmost path (resp., of the rightmost path) of $T(v)$ have slope $\frac{5\pi}{4}$ (resp., slope $\frac{7\pi}{4}$). 
\end{enumerate}
\end{lemma}



\begin{proof}
Our proof is by induction on the height $h$ of $T$. Recall that $T$ contains at least one internal vertex, hence $h\geq 1$. In the base case, $h=1$, that is, $T$ is a star. Let $v_1,\dots,v_k$ be the children of $r(T)$ in left-to-right order and note that $k\geq 2$ since the internal vertices of $T_G$ have degree at least $3$. Place $v_1,\dots,v_k$ at the points $(0,0),(2,0),\dots,(2k-2,0)$. Place $r(T)$ at $(k-1,k-1)$. Refer to \cref{fi:wheel}\blue{(bottom)}. The resulting straight-line drawing $\Gamma$ of $G[T]$ clearly satisfies \cref{prop:greedy,prop:convex-grid,prop:points,prop:paths}.


\begin{figure}[tb] 
	\centering
	\subfloat[]{
		\includegraphics[page=3,scale=0.8]{wheel.pdf}
		\label{fi:wheel}
	}\hfil
	\subfloat[]{
		\includegraphics[page=1,scale=0.5]{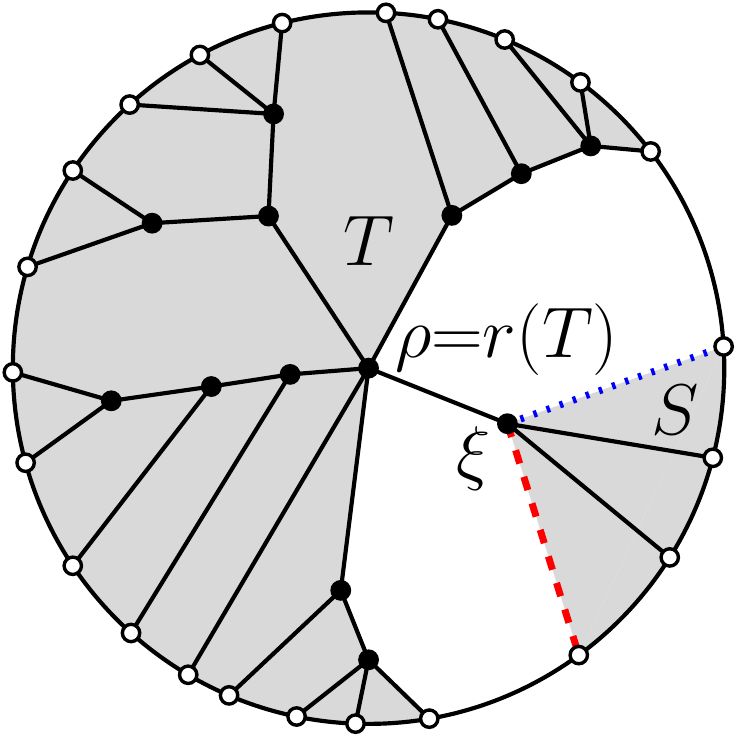}
		\label{fi:halin-a}
	}\hfil
	\subfloat[]{
		\includegraphics[page=2,scale=0.5]{halin-graph-new.pdf}
		\label{fi:halin-b}
	}
	\caption{
	(a) A convex angle-monotone drawing of a wheel on the grid (top) and the base case for the proof of \cref{le:halin-tree} (bottom). (b) The trees $T$ and $S$ for the proof of \cref{th:halin}. (c) The convex angle-monotone drawing $\Gamma_G$ of $G$ constructed from the drawings $\Gamma$ of $G[T]$ and $\overline{\Gamma_S}$ of $G[S]$.}
	\label{fi:halin-graph}
\end{figure}

Suppose now that $h>1$ and refer to \cref{fi:halin-b}. Let $T_1,\dots,T_k$ be the left-to-right order of the subtrees of $T$ rooted at the children of $r(T)$; as in the base case, we have $k\geq 2$. For each $T_i$ which is not a single vertex, assume to have inductively constructed a drawing $\Gamma_i$ of $G[T_i]$ satisfying \cref{prop:greedy,prop:convex-grid,prop:points,prop:paths}. For each $T_i$ which is a single vertex, let $\Gamma_i$ consist of the point $(0,0)$. For $i=1,\dots,k$, let $W_i$ be the width of~$\Gamma_i$. Place the drawings $\Gamma_1,\dots,\Gamma_k$ side by side, so that all their leaves lie on the $x$-axis, so that the leftmost leaf of $T_1$ is at $(0,0)$, and so that, for $i=1,\dots,k-1$, the rightmost leaf of $T_i$ is two units to the left of the leftmost leaf of $T_{i+1}$. We conclude the construction of  $\Gamma$ by placing $r(T)$ at $(\ell(T)-1,\ell(T)-1)$. We have the following.

\begin{claim} \label{cl:halin-subgraph}
$\Gamma$ satisfies \cref{prop:greedy,prop:convex-grid,prop:points,prop:paths}.
\end{claim}

\begin{proof}
\cref{prop:points} holds true since it is inductively satisfied by each drawing $\Gamma_i$ and since, by construction, the rightmost leaf of $T_i$ is two units to the left of the leftmost leaf of $T_{i+1}$, for $i=1,\dots,k-1$. 

Concerning \cref{prop:convex-grid}, we have that $W_\Gamma=\sum_{i=1}^k W_{\Gamma_i} + (k-1) =\sum_{i=1}^k (2\ell(T_i)-1) + (k-1)=2\ell(T)-1$, where we exploited $W_{\Gamma_i}=2\ell(T_i)-1$, which is true by induction. Further, by construction and by induction, each vertex of $T_i$ has a $y$-coordinate between $0$ and $\ell(T_i)-1$. Since $\ell(T_i)<\ell(T)$, the maximum $y$-coordinate of any vertex of $T$ in $\Gamma$ is the one of $r(T)$, hence $H_{\Gamma}=\ell(T)$.

\cref{prop:paths} holds true for each vertex different from $r(T)$ since it is inductively satisfied by each drawing $\Gamma_i$. Further, since $W_{\Gamma}=2\ell(T)-1$, since $r(T)$ lies at $(\ell(T)-1,\ell(T)-1)$, and since the leftmost and rightmost leaves of $T$ lie at $(0,0)$ and $(2\ell(T)-2,0)$, respectively, the slopes of the segments from $r(T)$ to such leaves are $\frac{5\pi}{4}$ and $\frac{7\pi}{4}$, respectively. This implies that the edges of the leftmost path (resp., of the rightmost path) of $T$ have slope $\frac{5\pi}{4}$ (resp., $\frac{7\pi}{4}$), given that the edges of the leftmost path of $T_1$ (resp.\, of the rightmost path of $T_k$) have slope $\frac{5\pi}{4}$ (resp., $\frac{7\pi}{4}$), by induction.

Finally, we prove \cref{prop:greedy}. We first prove that $\Gamma$ is convex. By induction, each internal face of $\Gamma$ which is also a face of $\Gamma_i$, with $i\in \{1,\dots,k\}$, is delimited by a convex polygon. The outer face of $\Gamma$ is delimited by a triangle, by \cref{prop:points,prop:paths}. It remains to prove that each internal face~$f$ incident to $r(T)$ is delimited by a convex polygon. Note that $f$ is delimited by the two edges $(r(T),r(T_i))$ and $(r(T),r(T_{i+1}))$, for some $i\in \{1,\dots,k-1\}$, by the rightmost path of $T_i$, by the leftmost path of $T_{i+1}$, and by the edge of $G[T]$ connecting the rightmost leaf of $T_i$ with the leftmost leaf of $T_{i+1}$. 

\begin{itemize}
	\item The angle of $f$ at $r(T)$ is at most $\frac{\pi}{2}$, by \cref{prop:paths}.
	\item The angles of $f$ at the internal vertices of the rightmost path of $T_i$ or of the leftmost path of $T_{i+1}$ are exactly $\pi$, by \cref{prop:paths}.
	\item The angle of $f$ at the rightmost leaf of $T_i$ (resp., at the leftmost leaf of $T_{i+1}$) is $\frac{3\pi}{4}$ if $T_i$ (resp., $T_{i+1}$) is not a single vertex or at most $\frac{3\pi}{4}$ otherwise, by \cref{prop:points,prop:paths}.
	\item The angle of $f$ at $r(T_i)$ is larger than or equal to $\frac{\pi}{2}$ and smaller than $\pi$; namely, the slope of the edge $(r(T),r(T_i))$ is in the interval $[\frac{5\pi}{4},\frac{7\pi}{4})$, by \cref{prop:paths} and by $i<k$; further, the slope of the edge of the rightmost path of $T_i$ incident to $r(T_i)$ is $\frac{7\pi}{4}$, by \cref{prop:paths}.
	\item Symmetrically, the angle of $f$ at $r(T_{i+1})$ is larger than or equal to $\frac{\pi}{2}$ and smaller than $\pi$.
\end{itemize}

We now prove that $\Gamma$ is angle-monotone. Let $u$ and $v$ be any two vertices of $T$. If $u$ and $v$ both belong to the same subtree $T_i$ of $T$, for some $i\in \{1,\dots,k\}$, then a $\beta$-monotone path between $u$ and $v$ exists in $\Gamma$ since it exists in $\Gamma_i$, by induction. Otherwise, either $u$ and $v$ belong to distinct subtrees $T_i$ and $T_j$ of $T$, or one of $u$ and $v$ is $r(T)$.

In the former case, suppose w.l.o.g.\ that $i<j$. Let $P$ be the path from $u$ to~$v$ consisting of: (i) the rightmost path $P_u$ of $T_i(u)$; (ii) the path $P_{uv}$ in $G[T]$ from the rightmost leaf of $T_i(u)$ to the leftmost leaf of $T_j(v)$ that only passes through leaves of $T$; and (iii) the leftmost path $P_v$ of $T_j(v)$. Since the edges of $P_u$ (which are traversed in the direction of $P_u$) and those of $P_v$ (which are traversed in the direction opposite to the one of $P_v$) have slope $\frac{7\pi}{4}$ and $\frac{\pi}{4}$, by \cref{prop:paths}, and the edges of $P_{uv}$ have slope $0$, by \cref{prop:points}, we have that $P$ is $0$-monotone.

In the latter case, suppose w.l.o.g.\ that $v=r(T)$ and that $u\in V(T_i)$, for some $i\in \{1,\dots,k\}$. By \cref{prop:paths}, all the edges of the path from $u$ to $v$ in $T$ have slope in the closed interval $[\frac{\pi}{4},\frac{3\pi}{4}]$, hence such a path is $\frac{\pi}{2}$-monotone. 

It follows that $\Gamma$ is angle-monotone; this concludes the proof of the claim.\end{proof}

\cref{cl:halin-subgraph} concludes the proof of the lemma.
\end{proof}

We are now ready to prove \cref{th:halin}. We construct a drawing $\Gamma_G$ of $G$ as follows; refer to \cref{fi:halin-b}. First, we initialize $\Gamma_G$ to the drawing $\Gamma$ of $G[T]$ obtained by applying \cref{le:halin-tree}. Further, we apply \cref{le:halin-tree} a second time in order to construct a drawing $\Gamma_S$ of $G[S]$. Let $\overline{\Gamma_S}$ be the drawing of $G[S]$ obtained by rotating $\Gamma_S$ by $\pi$ radians. 
We translate $\overline{\Gamma_S}$ so that $\xi$ lies one unit above~$\rho$. 
Further, we draw the edge $(\rho,\xi)$ as a vertical straight-line segment.
Finally, we draw the edge between the leftmost (rightmost) leaf of $S$ and the rightmost (leftmost) leaf of $T$ as a straight-line segment.

We have the following claim, which concludes the proof of \cref{th:halin}.

\begin{claim} \label{cl:halin-final}
	$\Gamma_G$ is a convex angle-monotone drawing of $G$ on an $O(n) \times O(n)$ grid. 
\end{claim}

\begin{proof}
	First, \cref{prop:convex-grid,prop:paths} of \cref{le:halin-tree} ensure that the width of $\Gamma_G$ is equal to $\max(2\ell(T)-1,2\ell(S)-1)$ and that the height of $\Gamma_G$ is equal to $\ell(T)+\ell(S)$. Both such values are in $O(n)$.
	
	Second, we prove that $\Gamma_G$ is convex. Every face of $\Gamma_G$ which is also a face of~$\Gamma$ or $\overline{\Gamma_S}$ is delimited by a convex polygon since $\Gamma$ and $\overline{\Gamma_S}$ are convex, by \cref{prop:greedy} of \cref{le:halin-tree}. Further, by \cref{prop:points,prop:paths} of \cref{le:halin-tree}, the outer face of $\Gamma_G$ is delimited by an isosceles trapezoid. Finally, consider any face $f$ incident to the edge $(\rho,\xi)$. By \cref{prop:paths} of \cref{le:halin-tree}, the angles of $f$ incident to the internal vertices of the leftmost and rightmost paths of $T$ are equal to $\pi$, hence $f$ is delimited by a quadrilateral $Q$; the angles of $f$ incident to $\rho$ and to $\xi$ are $\frac{3\pi}{4}$, again by \cref{prop:paths} of \cref{le:halin-tree} and since the edge $(\rho,\xi)$ is vertical, hence the remaining two angles of $Q$ sum up to $\frac{\pi}{2}$. It follows that $Q$ is convex.
	
	Third, we prove that $\Gamma_G$ is angle-monotone. Let $u$ and $v$ be any two vertices of $G$. If~$u$ and $v$ both belong to $T$ or both belong to $S$, then a $\beta$-monotone path between $u$ and $v$ exists in $\Gamma_G$ since it exists in $\Gamma$ or in $\overline{\Gamma_S}$, respectively, by \cref{le:halin-tree}. Otherwise, we can assume that $u$ belongs to $S$ and that $v$ belongs to $T$.
	Then the path $P$ from $u$ to $v$ in $T_G$ is $\frac{3\pi}{2}$-monotone. Namely, by \cref{prop:paths} of \cref{le:halin-tree}, the edge of $P$ in $S$, if any, has slope in the interval $[\frac{5\pi}{4},\frac{7\pi}{4}]$; further, by construction, the edge $(\xi,\rho)$ has slope $\frac{3\pi}{2}$; finally, again by \cref{prop:paths} of \cref{le:halin-tree}, all the edges of $P$ in $T$, if any, have slope in the interval $[\frac{5\pi}{4},\frac{7\pi}{4}]$. 
\end{proof}

\section{$\bf\alpha$-Schnyder Drawings of Plane Triangulations}\label{sse:schnyder}

In this section, we prove an exponential lower bound for the area requirements of $\alpha$-Schnyder drawings of plane triangulations, for any fixed $\alpha<\pi/3$. 
For a function $f(n)$ and a parameter $\epsilon >0$, we write $f(n) \in \Omega_\epsilon(n)$ if $f(n) \geq c_\epsilon n$ for some constant~$c_\epsilon>0$ which only depends on $\epsilon$.

\begin{figure}[tb]
	\subfloat[]{
		\includegraphics[scale=.55]{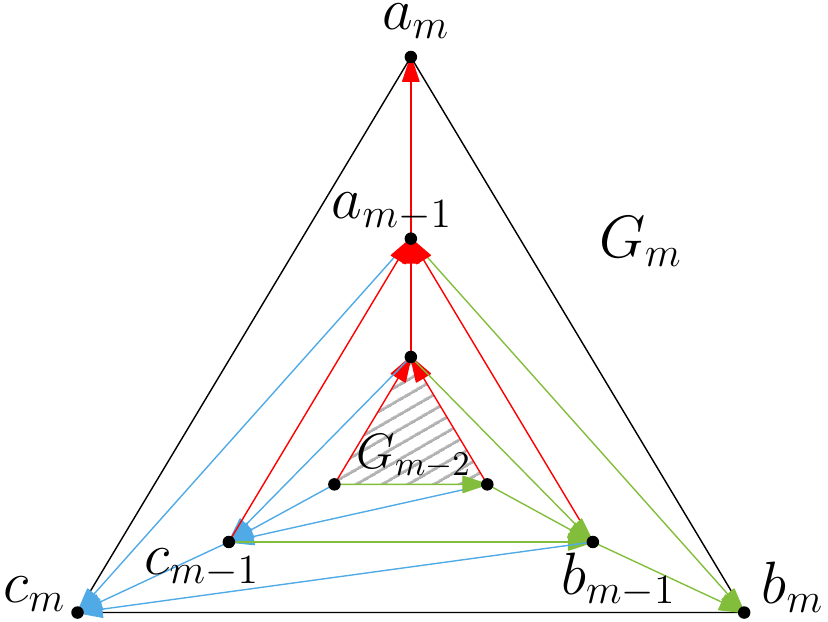}
		\label{fi:alpha-graphs}
	}
	\hfil
\subfloat[]{
	\includegraphics[scale=.5,page=2]{alpha-graphs.pdf}
	\label{fig:wedges}
}
	\caption{
		Illustrations for the proof of \cref{th:schnyder}:
		(a) The graph $G_m$.
		(b) The different angular widths of an $\alpha$-Schnyder drawing, with $\alpha < \frac{\pi}{3}$.}
	\label{fi:schnyder-lower}
\end{figure}

\begin{theorem} \label{th:schnyder}
There exists an infinite family $\mathcal F$ of bounded-degree planar $3$-trees such that, for any resolution rule, any \mbox{$n$-vertex} graph in $\mathcal F$ requires $2^{\Omega_\epsilon(n)}$ area in any \mbox{$(\frac{\pi}{3}$$-$$\epsilon)$-Schnyder} drawing, for any fixed $0 < \epsilon < \frac{\pi}{3}$.
\end{theorem} 

\begin{proof}
We start by defining, for each integer $m>0$, a $3m$-vertex plane \mbox{$3$-tree~$G_m$;} refer to \cref{fi:alpha-graphs}. The plane $3$-tree $G_1$ is a cycle $(a_1,b_1,c_1)$. For any integer $m>1$, the plane $3$-tree $G_{m}$ is obtained from the plane $3$-tree $G_{m-1}$ by embedding a cycle $(a_m,b_m,c_m)$ in the outer face of $G_{m-1}$, so that it contains $G_{m-1}$ in its interior, and by inserting the edges $(a_m,a_{m-1})$, $(b_m,a_{m-1})$, $(b_m,b_{m-1})$, $(c_m,a_{m-1})$, $(c_m,b_{m-1})$, and $(c_m,c_{m-1})$. 
	
Since $G_m$ is a plane $3$-tree, it has a unique Schnyder wood $(\mathcal T_1,\mathcal T_2,\mathcal T_3)$; see~\cite{b-or-00,DBLP:journals/combinatorics/FelsnerZ08}. Assume, w.l.o.g., that all the internal edges incident to $a_m$, to $b_m$, and to $c_m$ belong to $\mathcal T_1$, $\mathcal T_2$, and $\mathcal T_3$, respectively. We have the following.

\begin{claim} \label{cl:realizations} 
For every $1\leq k<m$, the Schnyder wood $(\mathcal T_1,\mathcal T_2,\mathcal T_3)$ of $G_m$ satisfies the following properties:

\begin{enumerate}[(a)]
	\item The edges $(a_k,b_k)$ and $(a_{k},c_k)$ belong to $\mathcal T_1$ and are directed towards~$a_k$ and the edge $(b_k,c_k)$ belongs to $\mathcal T_2$ and is directed towards $b_k$.
	\item The edge $(a_{k},a_{k+1})$ belongs to $\mathcal T_1$ and is directed towards~$a_{k+1}$, the edges $(b_{k+1},a_{k})$ and $(b_{k+1},b_{k})$ belong to $\mathcal T_2$ and are directed towards~$b_{k+1}$, and the edges $(c_{k+1},a_{k})$, $(c_{k+1},b_{k})$, and $(c_{k+1},c_{k})$ belong to $\mathcal T_3$ and are directed towards~$c_{k+1}$.
\end{enumerate}
\end{claim}

\begin{proof}
We prove the claim by reverse induction on $k$. If $k=m-1$, then Property~(2) of the definition of Schnyder wood directly implies that Property~(b) is satisfied. Further, by Property~(1) of the definition of Schnyder wood we have that all the edges (including $(a_k,b_k)$ and $(a_{k},c_k)$) that appear after  $(a_k,b_{k+1})$ and before $(a_k,c_{k+1})$ in clockwise order around $a_k$ belong to $\mathcal T_1$ and are directed towards~$a_{k}$; again by the same property, all the edges (including $(b_k,c_k)$) that appear after  $(b_k,c_{k+1})$ and before $(b_k,a_{k})$ in clockwise order around $b_k$ belong to $\mathcal T_2$ and are directed towards~$b_{k}$, hence Property~(a) is satisfied.

Now inductively assume that Properties~(a) and~(b) are satisfied for some $2\leq k\leq m-1$; we prove that they are satisfied for $k-1$, as well. By Property~(1) of the definition of Schnyder wood we have that the edge $(a_k,a_{k-1})$  belongs to $\mathcal T_1$ and is directed towards~$a_{k}$, since it appears after  $(a_k,b_k)$ and before $(a_k,c_{k})$ in clockwise order around $a_k$. 
By the same property, we have that the edges $(c_{k},a_{k-1})$, $(c_{k},b_{k-1})$, and $(c_{k},c_{k-1})$ belong to $\mathcal T_3$ and are directed towards~$c_{k}$, since they appear after $({c_{k},a_k})$ and before $({c_{k},b_k})$ in clockwise order around $c_k$.
Again by the same property, the edges $(b_{k},a_{k-1})$ and $(b_{k},b_{k-1})$ belong to $\mathcal T_2$ and are directed towards~$b_{k}$, since they appear after $(c_k,b_k)$ and before $(a_k,b_k)$ in clockwise order around $b_k$. This concludes the proof that Property~(b) is satisfied.
The proof that Property~(a) is satisfied is the same as in the base case.
\end{proof}

We now prove that, for any fixed $\alpha=\frac{\pi}{3} -\epsilon$ with $\epsilon >0$,  any $\alpha$-Schnyder drawing $\Gamma$ of $G_m$ (respecting the plane embedding of $G_m$) requires $2^{\Omega_\epsilon(m)}$ area. To this aim, we exploit the next claim.

\begin{restatable}{claim}{clTrianglesArea}\label{cl:triangles-area}
In $\Gamma$, the area of the triangle  $(a_i,b_i,c_i)$ is at least $k_\epsilon$ times the area of the triangle  $(a_{i-1},b_{i-1},c_{i-1})$, for any $i=2,\dots,m-1$, where $k_\epsilon>1$ is a constant only depending on $\epsilon$.
\end{restatable}

\begin{proof}
	Let $A_{i-1}$ and $A_i$ denote the areas of the triangles $\Delta_{i-1}=(a_{i-1},b_{i-1},c_{i-1})$ and $\Delta_i = (a_i,b_i,c_i)$ in $\Gamma$, respectively. 
	We prove that $\frac{A_i}{A_{i-1}} \geq k_\epsilon$, where  $k_\epsilon>1$ is a constant that depends only on $\epsilon$; refer to \cref{cl:triangles-area-a}.
	
	In the following, given a pair of elements $x$ and $y$, each representing either a point or a vertex in $\Gamma$, we will denote by $\overline{x y}$ and by $\ell(x,y)$ the segment connecting $x$ and $y$ in $\Gamma$, and the line passing through both $x$ and $y$ in $\Gamma$, respectively.
	
	The following statements exploit \cref{cl:realizations} and the fact that $\Gamma$ is an $\alpha$-Schnyder drawing (refer to \cref{fig:wedges}):
	\begin{itemize}
		\item the slopes of the edges $(b_i,a_i)$ and $(c_i,a_i)$ are in the range $[\frac{\pi}{2}-\frac{\alpha}{2},\frac{\pi}{2}+\frac{\alpha}{2}]$;
		\item the slope of the edge $(b_i,c_i)$ is in the range $[\frac{5\pi}{6}-\frac{\alpha}{2},\frac{5\pi}{6}+\frac{\alpha}{2}]$; and
		\item the slope of the edge $(c_i,a_{i-1})$ is in the range $[\frac{\pi}{6}-\frac{\alpha}{2},\frac{\pi}{6}+\frac{\alpha}{2}]$.
	\end{itemize}
	
	Next, we reduce the problem of proving $\frac{A_i}{A_{i-1}}\geq k_\epsilon$ to the one of determining a lower bound for the ratio between the areas of two triangles whose sides (except for one) have fixed slopes. 
	
	Let $o'$ be the intersection point between the edge $(a_i,b_i)$ and the line $\ell(c_{i},a_{i-1})$. 
	Let $A'$ denote the area of the triangle $\Delta' = (c_i,o',b_i)$; see \cref{cl:triangles-area-b}. Observe that $\Delta'$ strictly encloses $\Delta_{i-1}$; therefore, we have that $A'  > A_{i-1}$. 
	It follows that in order to prove that $\frac{A_i}{A_{i-1}} \geq k_\epsilon$, it suffices to prove that $\frac{A_i}{A'} \geq k_\epsilon$.

	Let $o''$ be the intersection point between the edge $(a_i,b_i)$ and the line with slope $\frac{\pi}{6}+\frac{\alpha}{2}$ passing through $c_i$. 
	Let $A''$ denote the area of the triangle $\Delta''=(c_i,o'',b_i)$; see \cref{cl:triangles-area-c}. Since the slope of the segment $\overline{c_i o'}$ is in the range $[\frac{\pi}{6}-\frac{\alpha}{2},\frac{\pi}{6}+\frac{\alpha}{2}]$ (as it contains the drawing of the edge $(c_i,a_{i-1})$), we have that $\Delta''$ encloses $\Delta'$; therefore, $A''  \geq A'$. 
	It follows that in order to prove that $\frac{A_i}{A'} \geq k_\epsilon$, it suffices to prove that $\frac{A_i}{A''} \geq k_\epsilon$.

	Let $o'''$ be the intersection point between the edge $(a_i,b_i)$ and the line with slope $\frac{\pi}{2}-\frac{\alpha}{2}$ passing through $c_i$. 
	Let $A'''$ denote the area of the triangle $\Delta'''=(c_i,o''',o'')$; see \cref{cl:triangles-area-c}.
	Since the slope of the edge $(c_i,a_i)$ is in the range $[\frac{\pi}{2}-\frac{\alpha}{2},\frac{\pi}{2}+\frac{\alpha}{2}]$, we have that $\Delta'''$ is enclosed in the triangle $\Delta^+ = (c_i,a_i,o'')$. Let $A^+$ denote the area of the triangle $\Delta^+$.
	We have that $A''' \leq A^+ = A_i - A''$. Thus, $A_i \geq A'' + A'''$.
	It follows that in order to prove that $\frac{A_i}{A''} \geq k_\epsilon$, it suffices to prove that $\frac{A''+ A'''}{A''} \geq k_\epsilon$, i.e., 	$\frac{A'''}{A''} \geq k_\epsilon -1$.

	Let $o^*$ be the intersection point between the line $\ell(a_i,b_i)$ and the line with slope $\frac{5\pi}{6}-\frac{\alpha}{2}$ passing through $c_i$; see \cref{cl:triangles-area-d}.
	Let $A^*$ denote the area of the triangle $\Delta^*=(c_i,o'',o^*)$. Since the slope of the edge $(c_i,b_i)$ is in the range $[\frac{11\pi}{6}-\frac{\alpha}{2},\frac{11\pi}{6}+\frac{\alpha}{2}]$, we have that $\Delta^*$ encloses $\Delta''$; therefore, we have that $A^*  \geq A''$. It follows that in order to prove that 
	$\frac{A'''}{A''} \geq k_\epsilon -1$, it suffices to prove that
	$\frac{A'''}{A^*} \geq k_\epsilon-1$.
	
	Finally, let $o^\diamond$ be the intersection point between the line $\ell(c_i, o''')$ and the line passing through $o''$ and perpendicular to $\ell(c_i,o'')$. Consider the right triangle $\Delta^\diamond = (c_i,o^\diamond,o'')$; see \cref{cl:triangles-area-d}. We claim that $\Delta'''$ strictly encloses  $\Delta^\diamond$, and thus, $A^\diamond < A'''$, which implies that in order to prove that 
	$\frac{A'''}{A^*} \geq k_\epsilon-1$, it suffices to prove the following main inequality
	\begin{equation}\label{eq:main}
	\frac{A^\diamond}{A^*} \geq k_\epsilon-1.
	\end{equation}
	To show that $\Delta'''$ strictly encloses  $\Delta^\diamond$, we prove that the internal angle of $\Delta'''$ at $o''$ is larger than $\frac{\pi}{2}$. This immediately descends from the fact that the slope of the segment $\overline{o''a_i }$ is in the range $[\frac{\pi}{2}-\frac{\alpha}{2},\frac{\pi}{2}+\frac{\alpha}{2}]$ (as it overlaps with the edge $(b_i,a_i)$ which belongs to $\mathcal T_1$), that the slope of the segment $\overline{o'' c_i }$ is $\frac{7\pi}{6}+\frac{\alpha}{2}$, by construction, and that $\alpha < \frac{\pi}{3}$.

	We are now ready to compute $A^\diamond$ and give an upper bound for $A^*$. In the following, we denote by $h$ the length of the segment $\overline{c_i o''}$; refer to \cref{cl:triangles-area-d}.
	
	\paragraph{Value of $A^\diamond$.}
	Let us denote by $\delta$ the interior angle of $\Delta^\diamond$ at $c_i$. Since, by construction, the slope of the segments $\overline{c_i o^\diamond}$ and $\overline{c_i o''}$ are $\frac{\pi}{2}-\frac{\alpha}{2}$ and $\frac{\pi}{6}+\frac{\alpha}{2}$, respectively, we have that $\delta = \frac{\pi}{3}$. Therefore,
	since $\Delta^\diamond$ is a right triangle whose cathetus incident to $\delta$ is the segment $\overline{c_i o''}$, the following holds
	
	\begin{equation}\label{eq:A-triple}
	A^\diamond = \frac{h^2 \tan(\delta)}{2} = \frac{\sqrt{3}}{2}h^2.
	\end{equation}
	
	\paragraph{Upper bound for $A^*$.}
	Let $\sigma$, $\beta$, and $\gamma$ be the internal angles of the triangle $\Delta^*$ at $c_i$, $o''$, and $o^*$, respectively.
	First, we show that
	\begin{inparaenum}[(i)]
		\item $\sigma = \frac{\pi}{3}+\alpha$,
		\item $\beta \leq \frac{\pi}{3}$, and
		\item $\gamma \geq \frac{\pi}{3}-\alpha$.
	\end{inparaenum} 

	Item (i) is due to the facts:
	\begin{inparaenum}
		\item the angle $\sigma'$ spanned by a clockwise rotation around~$c_i$ bringing a ray originating at~$c_i$ and directed rightward to overlap with $\overline{c_i o^*}$ is $\frac{\pi}{6}+\frac{\alpha}{2}$, given that the slope of $\overline{c_i o^*}$ is $\frac{11\pi}{6}-\frac{\alpha}{2}$, by construction;
		\item the angle $\sigma''$ spanned by a counter-clockwise rotation around~$c_i$ bringing a ray originating at~$c_i$ and directed rightward to overlap with $\overline{c_i o''}$ is $\frac{\pi}{6}+\frac{\alpha}{2}$, given that the slope of $\overline{c_i o''}$ is $\frac{\pi}{6}+\frac{\alpha}{2}$, by construction; and 
		\item $\sigma=\sigma' + \sigma''$.
	\end{inparaenum}

	Item (ii) is due to the following facts:
	\begin{inparaenum}
		\item the angle $\beta'$ spanned by a counter-clockwise rotation around~$o''$ bringing a ray originating at~$o''$ and directed rightward to overlap with $\overline{o'' o^*}$ is at most $\frac{3\pi}{2}+\frac{\alpha}{2}$, given that the slope of $\overline{o'' o^*}$ is in the range $[\frac{3\pi}{2}-\frac{\alpha}{2},\frac{3\pi}{2}+\frac{\alpha}{2}]$ (as it overlaps with the edge $(b_i,a_i)$, which belongs to $\mathcal T_1$);
		\item the angle $\beta''$ spanned by a counter-clockwise rotation around~$o''$ bringing a ray originating at~$o''$ and directed rightward to overlap with $\overline{o'' c_i}$ is $\frac{7\pi}{6}+\frac{\alpha}{2}$, given that the slope of $\overline{o'' c_i}$ is $\frac{7\pi}{6}+\frac{\alpha}{2}$, by construction; and
		\item $\beta = \beta' - \beta''$.
	\end{inparaenum}

	Item (iii) is due to (i), (ii) and of $\sigma+\beta+\gamma=\pi$.
	
	Let $\kappa$ denote the length of the segment $\overline{c_i o^*}$. By the {\em law of sines} applied to the triangle $\Delta^*$, we have that
	$\frac{h}{\sin(\gamma)} = \frac{\kappa}{\sin(\beta)}$, i.e., $\kappa = h \frac{\sin(\beta)}{\sin(\gamma)}$. Applying Items (ii) and (iii), the following holds
	\begin{equation}\label{eq:kappa}
	\kappa \leq h \frac{\sin(\frac{\pi}{3})}{\sin(\frac{\pi}{3}-\alpha)}= h\frac{\sqrt{3}}{2\sin(\frac{\pi}{3}-\alpha)}.
	\end{equation}
	
	Since the segments $\overline{c_i o''}$ and $\overline{c_i o^*}$ form the two sides of $\Delta^*$ whose angle at~$c_i$~is~$\sigma$, we have that $A^* = \frac{h \cdot \kappa \sin(\sigma)}{2}$. Applying \cref{eq:kappa} and Item~(i), we thus have that the following holds
	\begin{equation}\label{eq:A-star}
	A^* \leq \frac{h^2\sqrt{3}}{4} \frac{\sin(\frac{\pi}{3}+\alpha)}{\sin(\frac{\pi}{3}-\alpha)}.
	\end{equation}
	
	\paragraph{Determining $k_\epsilon$.} 
	Combining \cref{eq:A-triple,eq:A-star}, we have that
	$$\frac{A^\diamond}{A^*} \geq 
	\frac{2\sin(\frac{\pi}{3}-\alpha)}{\sin(\frac{\pi}{3}+\alpha)}.
	$$
	By the above inequality, and exploiting the fact that $\frac{\pi}{3}-\alpha = \epsilon$ and $\frac{\pi}{3}+\alpha = \frac{2\pi}{3}-\epsilon$, we 
	have that \cref{eq:main} is satisfied by setting
	\begin{equation}\label{eq:k-epsilon}
	k_\epsilon = 1 + \frac{2\sin(\epsilon)}{\sin(\frac{2\pi}{3}-\epsilon)}.
	\end{equation}
	
	Note that $k_\epsilon$ is a constant, for any fixed $\epsilon$; further, $k_\epsilon>1$, since $\frac{2\sin(\epsilon)}{\sin(\frac{2\pi}{3}-\epsilon)}>0$ with $0 < \epsilon < \frac{\pi}{3}$. This concludes the proof. \end{proof}

\cref{cl:triangles-area} immediately implies that the area of the triangle $(a_{m-1},b_{m-1},c_{m-1})$ is at least $k_\epsilon^{m-2}$ times the area of the triangle $(a_1,b_1,c_1)$. Since the area of the triangle $(a_1,b_1,c_1)$ is greater than some constant depending on the adopted resolution rule, we get that the area of $\Gamma$ is in $2^{\Omega_\epsilon(m)}$.

We are now ready to define the family $\mathcal F$ of the statement. For any positive integer $m$ and any $n=6m-2$, we construct the graph $F_n \in \mathcal F$ from the complete graph $K_4$ on $4$ vertices, by taking two copies $G'_m$ and $G''_m$ of $G_m$ and by identifying the vertices incident to the outer face of each copy with the three vertices incident to two distinct triangular faces of~the~$K_4$. Observe that $F_n$ is a bounded-degree planar $3$-tree. 
In any $\alpha$-Schnyder drawing $\Gamma$ of $F_n$ (in fact in planar drawing of $F_n$), at least one of the two copies of $G_m$, say $G'_m$, is drawn so that its outer face is delimited by the triangle $(a_m,b_m,c_m)$. Since $F_n$ has a unique Schnyder wood~\cite{b-or-00,DBLP:journals/combinatorics/FelsnerZ08}, the restriction of such a Schnyder wood to the internal edges of $G'_m$ satisfies the properties of \cref{cl:realizations}. It follows that the restriction of $\Gamma$ to $G'_m$ is an $\alpha$-Schnyder drawing of $G'_m$ (respecting the plane embedding of $G'_m$), and therefore it requires $2^{\Omega_\epsilon(m)}$ area. The proof is concluded by observing that $m \in \Omega(n)$. 
\end{proof}

	\begin{figure}[ht]
		\centering
		\subfloat[\label{cl:triangles-area-a}]{
			\includegraphics[scale=.65,page=4]{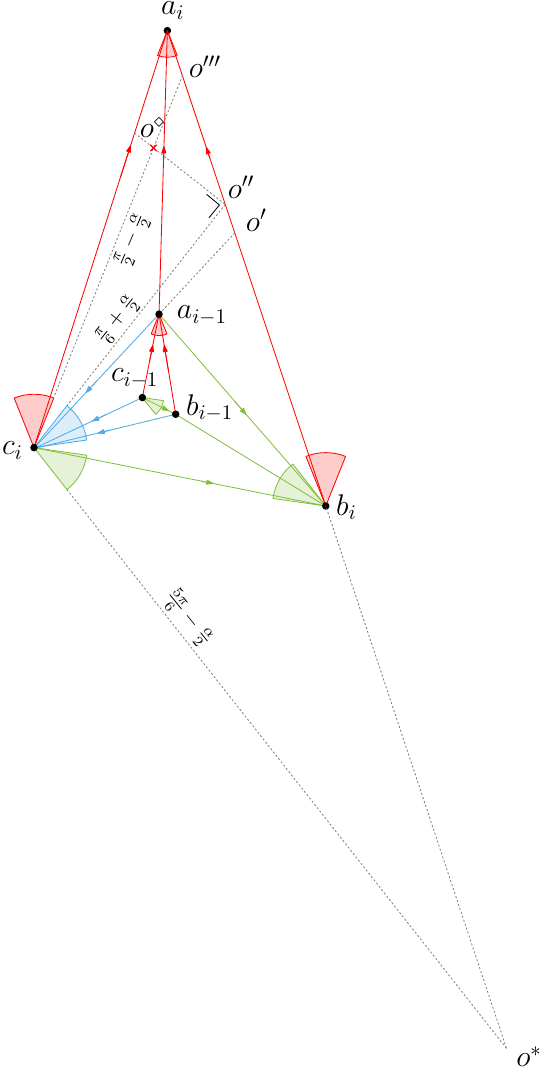}
		}
		\hfil
		\subfloat[\label{cl:triangles-area-b}]{
			\includegraphics[scale=.65,page=5]{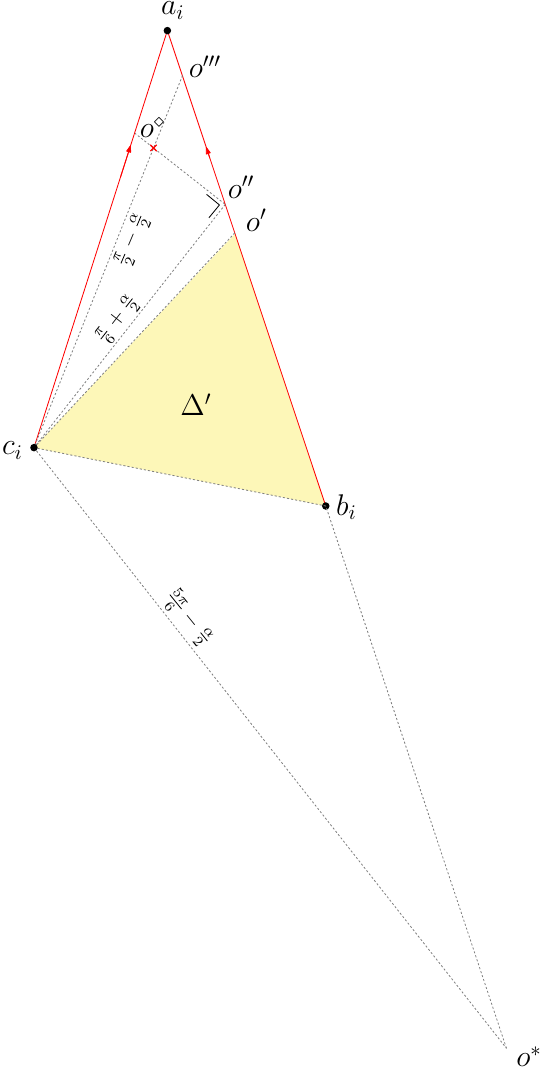}
		}
		\\
		\subfloat[\label{cl:triangles-area-c}]{
			\includegraphics[scale=.65,page=6]{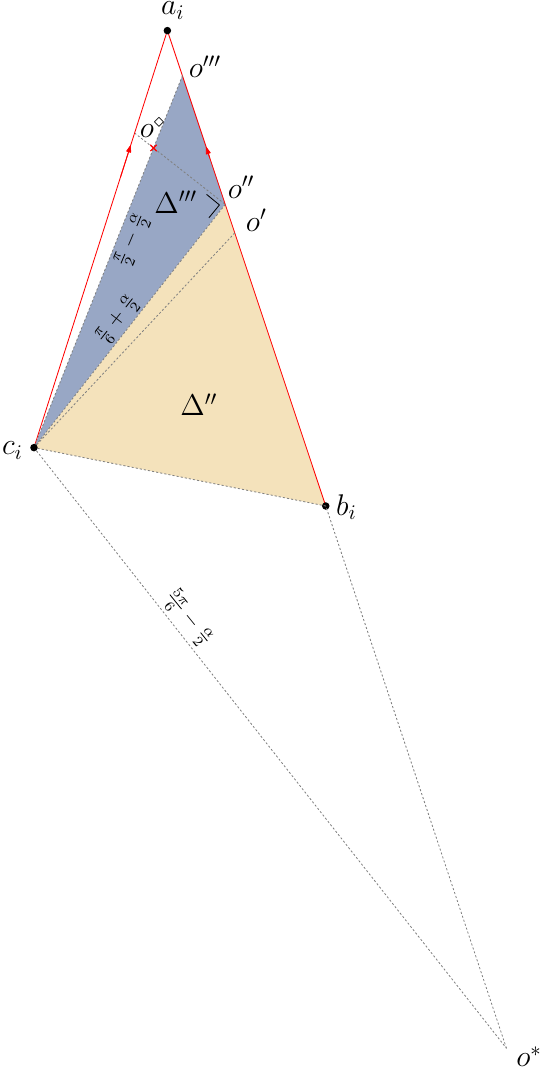}
		}
		\hfil
		\subfloat[\label{cl:triangles-area-d}]{
			\includegraphics[scale=.65,page=8]{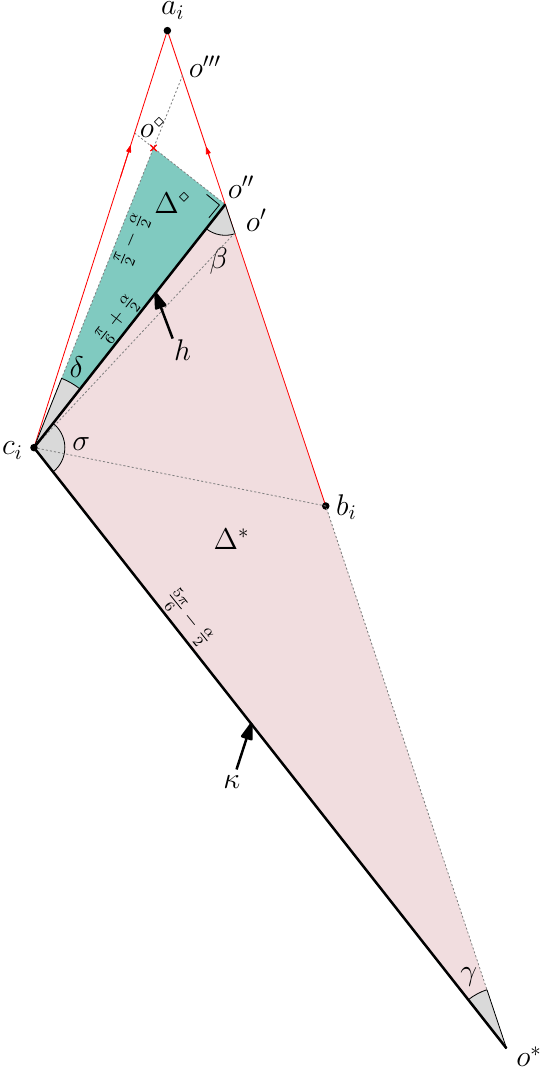}
		}
		\caption{Illustration for the proof of \cref{cl:triangles-area}. Values in radiants indicate segments' slopes.		
		}
		\label{fig:proof-of-triangles-area}
	\end{figure}

\section{Conclusions and Open Problems} \label{se:conclusions}  

In this paper, we refuted a claim by Cao et al.~\cite{css-osggrep-09} and re-opened the question of whether $3$-connected planar graphs admit planar, and possibly convex, greedy drawings on a polynomial-size grid. Further, we provided some evidence for a positive answer by showing that every $n$-vertex Halin graph admits a convex greedy drawing on an $O(n)\times O(n)$ grid; in fact, our drawings are angle-monotone, which is a stronger property than greediness. Moreover, we proved that $\alpha$-Schnyder drawings, which are an even more constrained drawing standard, might require exponential area for any fixed $\alpha<\frac{\pi}{3}$. 

Several questions remain open in this topic. We mention two of them that seem to be natural next steps. {\bf (Q1)} Does every $2$-outerplanar graph admit a planar, and possibly convex, greedy drawing on a polynomial-size grid? Note that the class of $2$-outerplanar graphs is strictly larger than the one of Halin graphs. {\bf (Q2)} Does every plane $3$-tree admit a planar greedy drawing on a polynomial-size grid? We indeed proved a negative answer if ``greedy drawing'' \mbox{is replaced by ``$\alpha$-Schnyder drawing'', for any fixed $\alpha<\frac{\pi}{3}$.}
	
\bibliographystyle{splncs04}
\bibliography{bibliography}
\end{document}